\documentclass{article}
\usepackage{amsmath,graphicx}

\usepackage{amsthm}
\usepackage{amssymb}
\usepackage{xypic}
\theoremstyle{plain}

\newtheorem{thm}{Theorem}
\newtheorem{proposition}[thm]{Proposition}
\newtheorem{cor}[thm]{Corollary}
\newtheorem{lem}[thm]{Lemma}

\theoremstyle{definition}
\newtheorem{df}[thm]{Definition}

\newtheorem{eg}[thm]{Example}

\def\shf{\mathcal}

\def\st{\textrm{star }}
\def\cl{\textrm{cl }}
\def\active{\textrm{active}}
\def\roi{\textrm{roi }}

\title{Modeling wireless network routing using sheaves}
\author{Michael Robinson\\
Mathematics and Statistics\\
American University\\
Washington, DC, USA\\
michaelr@american.edu}

\begin{document}
\maketitle
\begin{abstract}
This article explains how to construct a sheaf model for passing traffic through a wireless network with a single channel carrier sense multiple access/collision detection (CSMA/CD) media access model.
\end{abstract}
%
%
\section{Introduction}

We make the following \emph{single channel assumption}: if a link connected to a node is jammed, then that node cannot receive transmissions from \emph{any} other node.  

\begin{df}
An \emph{abstract simplicial complex} $X$ on a set $A$ is a collection of ordered subsets of $A$ that is closed under the operation of taking subsets.  We call an element of $X$ which itself contains $k+1$ elements a \emph{$k$-cell}.  We usually call a $0$-cell a \emph{vertex} and a $1$-cell an \emph{edge}.  

If $a,b$ are cells with $a \subset b$, we say that $a$ is a \emph{face} of $b$, and that $b$ is a \emph{coface} of $a$.  A cell of $X$ that has no cofaces is called a \emph{facet}.

The \emph{closure} $\cl Y$ of a set $Y$ of cells in $X$ is the smallest abstract simplicial complex that contains $Y$.  The \emph{star} $\st Y$ of a set $Y$ of cells in $X$ is the set of all cells that have at least one face in $Y$.  
\end{df}

Suppose a radio network consisting of a collection of nodes $N=\{n_i\}$ is active in a spatial region $R$.  Assume all nodes communicate through a single-channel, broadcast resource.  An open set $U_i \subset R$ is associated to each node $n_i$ that represents its \emph{transmitter coverage region}.  For each node $n_i$, a continuous function $s_i : U_i \to \mathbb{R}$ represents its \emph{signal level} at each point in $U_i$.  Without loss of generality, we assume that there is a global threshold $T$ for accurately decoding the transmission from any node.

\begin{df}
The \emph{link graph} is the following collection of subsets of $N$:
\begin{enumerate}
\item $\{n_i\} \in N$ for each node $n_i$, and
\item $\{n_i,n_j\}\in N$ if $s_i(n_j) > T$ and $s_j(n_i) >T$.
\end{enumerate}
The \emph{link complex} $L=L(N,U,s,T)$ is the clique complex of the link graph, which means that it contains all elements of the form $\{i_1,\dotsc,i_n\}$ whenever this set is a clique in the link graph.
\end{df}

\begin{proposition}
Each facet in the link complex is a maximal set of nodes that can communicate directly with one another (with only one transmitting at a time).
\end{proposition}
\begin{proof}
Let $c$ be a cell of the link complex.  By definition, for each pair of nodes, $i,j\in c$ implies that $s_i(n_j) > T$ and $s_j(n_i) >T$.  Therefore, $i$ and $j$ can communicate with one another.  
\end{proof}

\begin{cor}
Facets of the link complexes represent common broadcast resources.
\end{cor}

\section{Interference from a transmission}

The interference caused by a transmission impacts the usability of the network outside of the transmission's immediate vicinity.  This section builds a consistent definition of the \emph{region of influence} of a node or a link within the network.  To justify this definition, we use a local model that describes which configurations of nodes can transmit simultaneously.

\begin{df}
Suppose that $X$ is a simplicial complex (such as an interference or link complex) whose set of vertices is $N$.  Consider the following assignment $\shf{A}$ of additional information to capture which nodes are transmitting and decodable:
\begin{enumerate}
\item To each cell $c\in X$, assign the set
\begin{eqnarray*}
\shf{A}(c)&=&\{n \in N : \text{there exists a cell } d\in X \text{ with }\\
&& c \subset d \text{ and } n\in d\}\cup\{\perp\}
\end{eqnarray*} 
of nodes that have a coface in common with $c$, along with the symbol $\perp$.  We call $\shf{A}(c)$ the \emph{stalk} of $\shf{A}$ at $c$.
\item To each pair $c \subset d$ of cells, assign the \emph{restriction function}
\begin{equation*}
\shf{A}(c\subset d)(n) =
\begin{cases}
n&\text{if }n\in \shf{A}(d)\\
\perp&\text{otherwise}\\
\end{cases}
\end{equation*}
\end{enumerate}
\end{df}

For instance, if $c \in X$ is a cell of a link complex, $\shf{A}(c)$ specifies which nearby node is transmitting and decodable, or $\perp$ if none are.  The restriction functions relate the decodable transmitting nodes at the nodes to which nodes are decodable along an attached wireless link.  Similarly, if $c \in X$ is a cell of an interference complex, $\shf{A}(c)$ also specifies which nearby node is transmitting, and effectively locks out any interfering transmissions from other nodes.  

\begin{df}
The assignment $\shf{A}$ is called the \emph{activation sheaf} and is an example of a \emph{cellular sheaf} -- a mathematical object that stores local data.  The theory of sheaves explains how to extract consistent information, which in the case of networks consists of nodes that whose transmissions do not interfere with one another.

A \emph{section} of $\shf{A}$ supported on a subset $Y \subseteq X$ is an assignment $s:Y \to N$ so that for each $c \subset d$ in $Y$, $s(c) \in \shf{A}(c)$ and $\shf{A}(c\subset d)\left( s(c)\right) = s(d)$.  A section supported on $X$ is called a \emph{global section}.  
\end{df}

Specifically, global sections are complete lists of nodes that can be transmitting without interference.  

\begin{figure}
\begin{center}
\includegraphics[width=3in]{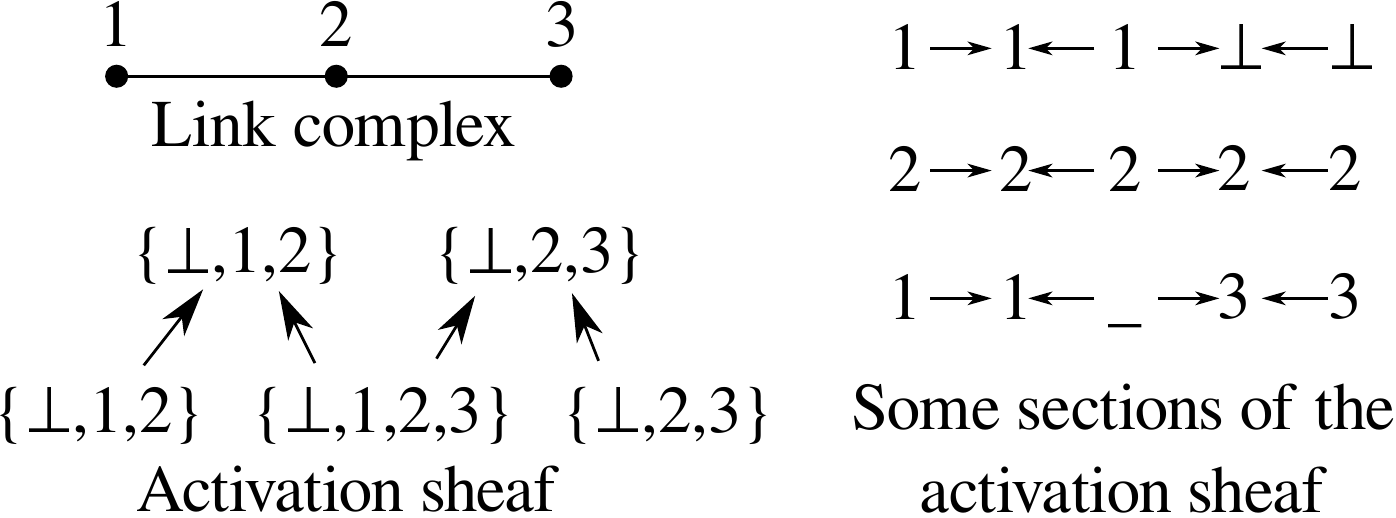}
\caption{A link complex (left top), sheaf $\shf{A}$ (left bottom), and three sections (right).  The restrictions are shown with arrows.  global section when node 1 transmits (right top), global section when node 2 transmits (right middle), and a local section with nodes 1 and 3 attempting to transmit, interfering at node 2 (right bottom)}
\label{fig:linesec}
\end{center}
\end{figure}

\begin{eg}
\label{eg:linesec}
Figure \ref{fig:linesec} shows a network with three nodes, labeled 1, 2, and 3.  When node 1 transmits, node 2 receives.  Because node 2 is busy, its link to node 3 must remain inactive (right top).  When node 2 transmits, both nodes 1 and 3 receive (right middle).  The right bottom diagram shows a local section that cannot be extended to the cell marked with a blank.  This corresponds to the situation where nodes 1 and 3 attempt to transmit but instead cause interference at node 2. 
\end{eg}

\begin{df}
Suppose that $s$ is a global section of $\shf{A}$.  The \emph{active region} associated to a node $n\in X$ in $s$ is the set
\begin{equation*}
\active (s,n) = \{a \in X : s(a)=n\},
\end{equation*}
which is the set of all nodes that are currently waiting on $n$ to finish transmitting.
\end{df}

\begin{lem}
\label{lem:act}
The active region of a node is a connected, closed subcomplex of $X$ that contains $n$.
\end{lem}
\begin{proof}
Consider a cell $c \in \active (s,n)$.  If $c$ is not a vertex, then there exists a $b \subset c$; we must show that $b\in \active (s,n)$.  Since $s$ is a global section $\shf{A}(b\subset c)s(b)=s(c)=n$.  Because $s(c) \not= \perp$, the definition of the restriction function $\shf{A}(b\subset c)$ implies that $s(b)=n$.  Thus $b\in \active (s,n)$ so $\active (s,n)$ is closed.

If $c \in \active (s,n)$, then $c$ and $n$ have a coface $d$ in common.  Since $s$ is a global section $s(d)=\shf{A}(c \subset d)s(c)=\shf{A}(c \subset d)n=n$.  Thus, $n \in \active(s,n)$, because $n$ is a face of $d$ and $\active (s,n)$ is closed.  This also shows that every cell in $\active (s,n)$ is connected to $n$.
\end{proof}

\begin{lem}
\label{lem:stactive}
The star over the active region of a node does not intersect the active region of any other node.
\end{lem}
\begin{proof}
Let $c\in \st \active (s,n)$.  Without loss of generality, assume that $c \notin \active (s,n)$.  Therefore, there is a $b \in \active (s,n)$ with $b\subset c$.  By the definition of the restriction function $\shf{A}(b\subset c)$, the assumption that $c\notin \active (s,n)$, and the fact that $s$ is a global section, $s(c)$ must be $\perp$.
\end{proof}

\begin{cor}
  If $s$ is a global section of an activation sheaf $\shf{A}$, then the \emph{support} of $s$ -- the set of cells $c$ where $s(c) \not= \perp$ -- consists of a disjoint union of active regions of nodes. 
\end{cor}

\begin{lem}
  \label{lem:activeinvariant}
The active region of a node is independent of the global section.  More precisely, if $r$ and $s$ are global sections of $\shf{A}$ and the active regions associated to $n \in X$ are nonempty in both, then $\active (s,n)=\active (r,n)$.
\end{lem}
\begin{proof}
Without loss of generality, we need only show that $\active (s,n) \subseteq \active (r,n)$.  If $c \in\active(s,n)$, there must be a cell $d\in X$ that has both $n$ and $c$ as faces.  Now $s(n)=r(n)=n$ by Lemma \ref{lem:act}, which means that $r(d)=\shf{A}(n \subset d)r(n)=n$.  Therefore, since $\active (r,n)$ is closed, this implies that $c \in \active(r,n)$.  
\end{proof}

\begin{cor}
  \label{cor:activationsections}
  The space of global sections of an activation sheaf consists of all sets of nodes that can be transmitting simultaneously without interference.
\end{cor}

\begin{df}
Because of the Lemmas, we call the star over an active region associated to a node $n$ the \emph{region of influence}.  The region of influence of a node can be written as a union
\begin{equation*}
\roi n = \bigcup_{f \in F} \st \cl f,
\end{equation*}
for a collection $F$ of facets, so we call the star over a closure of a facet the \emph{region of influence} of that facet.
\end{df}

\begin{cor}
\label{cor:unaffected}
The complement of the region of influence of a facet is a closed subcomplex.
\end{cor}

Although the space of global sections for an activation sheaf is a useful invariant, the cohomology of an activation sheaf is rather uninteresting.  However, activation sheaves are not sheaves of vector spaces, so we need to enrich their structure somewhat to see this.

\begin{df}
  If $\shf{A}$ is an activation sheaf on an abstract simplicial complex $X$, the \emph{vector activation sheaf} $\widehat{\shf{A}}$ is given by specifying its stalks and restrictions:
\begin{enumerate}
\item To each cell $c\in X$, let $\widehat{\shf{A}}(c)$ be the vector space whose basis is $\shf{A}\backslash\{\perp\}$ (so the dimension of this vector space is the cardinality of $\shf{A}$ without counting $\perp$)
\item The restriction map $\widehat{\shf{A}}(c\subset d)(n)$ is the basis projection, which is well-defined since $\shf{A}(d) \subseteq \shf{A}(c)$.
\end{enumerate}
\end{df}

\begin{thm}
  The dimension of the cohomology spaces of a vector activation sheaf $\widehat{\shf{A}}$ on a link complex $X$ are
  \begin{equation*}
    \text{dim }H^k(\widehat{\shf{A}}) = \begin{cases}
      \text{the total number of nodes}&\text{if }k = 0\\
      0&\text{otherwise}\\
      \end{cases}
  \end{equation*}
\end{thm}
\begin{proof}
  Every global section of $\shf{A}$ corresponds to a global section of $\widehat{\shf{A}}$, but formal linear combinations of global sections of $\shf{A}$ are also global sections of $\widehat{\shf{A}}$.  Therefore, a global section of $\widehat{\shf{A}}$ merely consists of a list of those nodes that are transmitting, without regard for whether they interfere.

  The fact that the other cohomology spaces are trivial is considerably more subtle.  Consider the decomposition
  \begin{equation*}
    X = \bigcup_i F_i
  \end{equation*}
  of the link complex into the set of its facets.  Suppose that $F_i$ is a facet of dimension $k$, and defining $\shf{F}_i$ to be the direct sum of $k+1$ copies of the constant sheaf supported on $F_i$.  (Each copy corresponds one of the vertices of $F_i$.)  Then there is an exact sequence of sheaves
  \begin{equation*}
    \xymatrix{
      0 \to \widehat{\shf{A}} \ar[r]^\Delta & \bigoplus_i \shf{F}_i \ar[r]^m & \shf{S} \to 0\\
      }
  \end{equation*}
  where $\Delta$ is a map that takes a basis vector corresponding to a given node to the linear combination of all corresponding basis vectors in each copy of the constant sheaves, and $m$ is therefore a kind of difference map.  This exact sequence leads to a long exact sequence
  \begin{equation*}
\dotsb H^{k-1}(\shf{S}) \to H^k(\widehat{\shf{A}}) \to \bigoplus_i H^k(\shf{F}_i) \to H^k(\shf{S}) \dotsb
  \end{equation*}
Since each $\shf{F}_i$ is a direct sum of constant sheaves supported on a closed subcomplex, it only has nontrivial cohomology in degree 0.

Observe that $\shf{S}$ is a sheaf supported on sets of cells lying in the intersections of facets.  By Corollary \ref{cor:unaffected}, $\shf{S}$ must be a direct sum of copies of constant sheaves supported on closed subcomplexes, like each $\shf{F}_i$.  Thus $\shf{S}$ only has nontrivial cohomology in degree 0, which means that for $k>1$, $H^k(\widehat{\shf{A}}) = 0$.

It therefore remains to address the $k=1$ case, which comes about from the exact sequence
\begin{equation*}
  \bigoplus_i H^0(\shf{F}_i) \to H^0(\shf{S}) \to H^1(\widehat{\shf{S}}) \to 0.
\end{equation*}
The leftmost map is surjective, since every global section of $\shf{S}$ is given by specifying a single transmitting node.  By picking exactly one facet containing that node, a global section of the corresponding $\shf{F}_i$ may be selected in the preimage.  Thus the map $H^0(\shf{S}) \to H^1(\widehat{\shf{S}})$ must be the zero map and yet also surjective.  This completes the proof.
\end{proof}

\section{Modeling link and node data payloads}

The activation sheaf describes the state of the network at a single instant in time.  Because the network conditions may change over time, the link and interference complexes may also change with time.  This section describes a general framework for representing both these changes and the data that is transmitted over the links.  

In order to capture changes in the network's topology over time, it is appropriate to use a single link or interference complex to represent a single timeslice.  To represent how the network's state evolves over several consecutive timeslices, we construct additional links between nodes in different timeslices.  These links carry information from one timeslice to the next.

Extending the definition of a link complex above, again suppose that a radio network consists of a collection of nodes $N=\{n_i\}$ in a spatial region $R$ in which a coverage region $U_i \subset R$ is associated to each node $n_i$.  For each node $n_i$, assign a signal level function $s_i : U_i \times \mathbb{Z} \to \mathbb{R}$, where the second input represents time.  Again, without loss of generality, we assume that there is a global decoding threshold $T$.

\begin{df}
The \emph{time-dependent link graph} is the following collection of subsets of $N \times \mathbb{Z}$:
\begin{enumerate}
\item $\{(n_i,t)\} \in N$ for each node $n_i$ and $t \in \mathbb{Z}$,
\item $\{(n_i,t),(n_j,t)\}\in N$ if $s_i(n_j,t) > T$ and $s_j(n_i,t) >T$, and
\item $\{(n_i,t),(n_i,t+1)\}$ for each node $n_i$.
\end{enumerate}
The \emph{time-depdendent link complex} $L=L(N,U,s,T)$ is the clique complex of the time-dependent link graph, which means that it contains all elements of the form $\{i_1,\dotsc,i_n\}$ whenever this set is a clique in the link graph.  The \emph{time $t$ timeslice} of $L$ is the maximal subcomplex of $L$ containing vertices from $N \times \{t\}$.  The \emph{time-dependent interference complex} and its timeslices can be defined in an analogous manner.
\end{df}

\begin{figure}
\begin{center}
\includegraphics[width=3in]{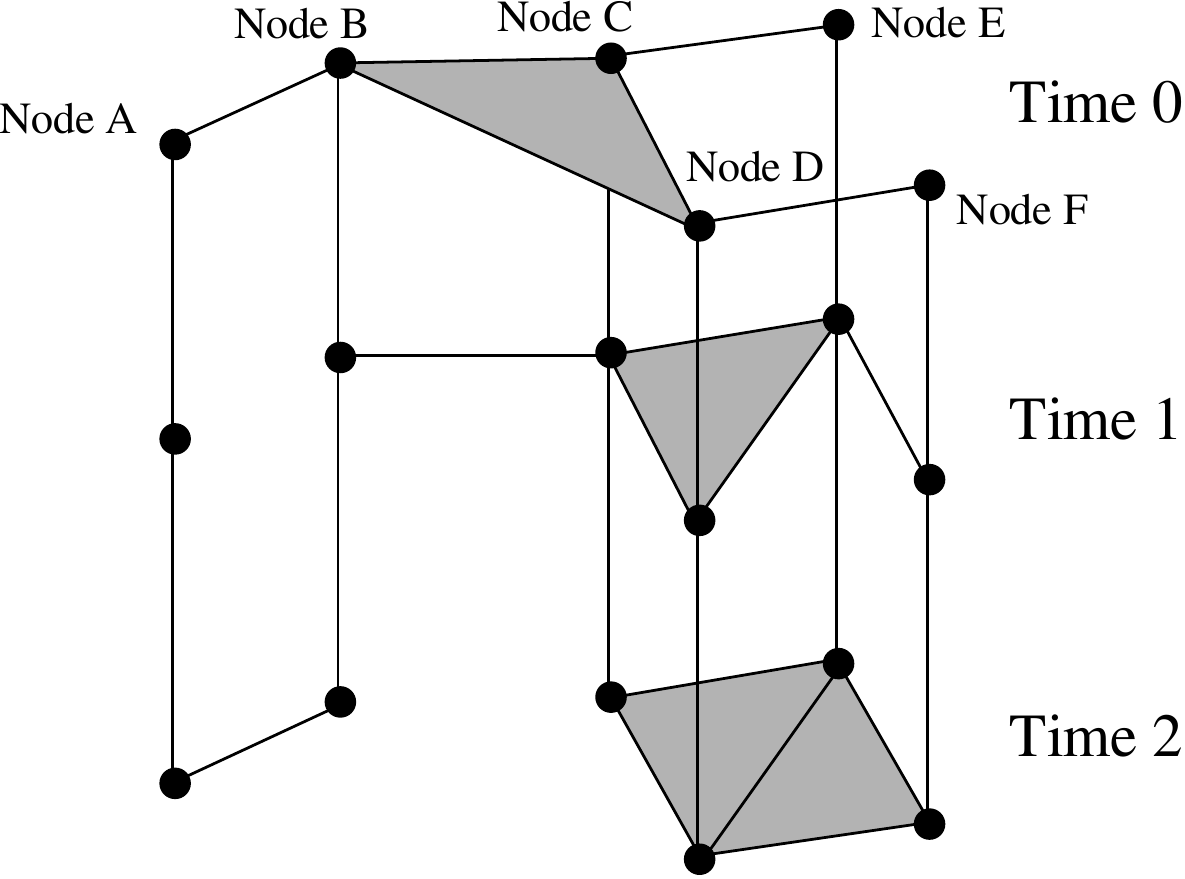}
\caption{Evolution of a simplicial complex model of a wireless network through time.  The nodes A--F listed at time 0 are repeated through the other timeslices as proceeding vertically downward through the diagram.}
\label{fig:timeslices}
\end{center}
\end{figure}

Figure \ref{fig:timeslices} shows an example of a time-dependent link complex.  Notice that each timeslice is a link complex, and that timeslices are attached to one another only by edges between consecutive copies of the same node.  The interpretation is that this represents a network in which the links are memoryless; only nodes can retain information from one timeslice to the next.

In order to represent the information that is retained by a node concretely, we will first construct a sheaf model of a queue and then generalize.  For example, consider the following diagram of sets
\begin{equation*}
\xymatrix{
\dotsb\ar[r]&0&A\ar[r]\ar[l]&0&A\ar[r]\ar[l]&\dotsb
}
\end{equation*} 
in which the arrows represent functions taking all elements of $A$ to $0$.  This diagram is a sheaf over a simplicial complex whose sections are given by assigning (possibly different) elements of $A$ to each vertex.  In other words, the space of sections of this sheaf is the space of $A$-valued sequences.  This sheaf should be interpreted as a $1$-deep queue: each $A$ in the diagram above represents the contents of the queue at a timeslice, but nothing is retained from one timeslice to the next.  To generalize this construction to larger queues, consider the following definition.

\begin{df}
Suppose $X$ is the abstract simplicial complex model of $\mathbb{R}$ whose vertices are given by the set of integers and whose edges are given by pairs $(t,t+1)$.  The \emph{$n$-term grouping sheaf $\shf{A}^{(n)}$} is given by the diagram written over $X$
\begin{equation*}
\xymatrix{
\ar[r]&A^{n-1}&A^n\ar[r]^{\sigma_+}\ar[l]^{\sigma_-}&A^{n-1}&A^n\ar[r]^{\sigma_+}\ar[l]^{\sigma_-}&
}
\end{equation*}
where $\sigma_\pm$ are the projection functions given by
\begin{equation*}
\sigma_-(x_1,\dotsc,x_n)=(x_2,\dotsc,x_n)
\end{equation*}
and
\begin{equation*}
\sigma_+(x_1,\dotsc,x_n)=(x_1,\dotsc,x_{n-1}).
\end{equation*}
\end{df}

The $\sigma_\pm$ projection functions in a grouping sheaf represent the action of advancing the queue from one timestamp to the next.  Specifically, consider the diagram shown in Figure \ref{fig:grouping_section}, which represents a $3$-term grouping sheaf.  To read the diagram as a $3$-deep queue, think of data as entering the queue from the first element of the vectors (the top) and exiting from the last (the bottom).  The action of $\sigma_-$ is simply to drop the first element, which maintains consistency between the current timestep and the one previous.  Conversely, $\sigma_+$ drops the last element, which links the current timestep to the next one.  

\begin{figure}
\begin{center}
\includegraphics[width=3in]{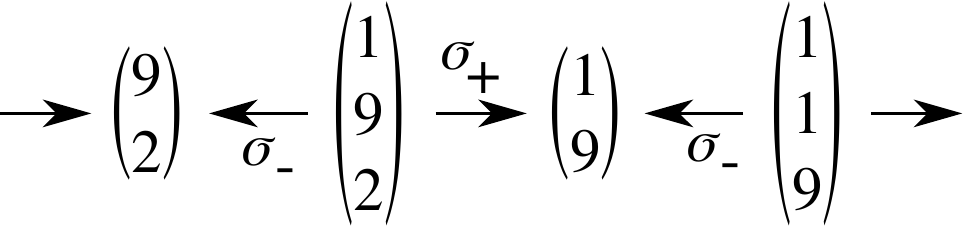}
\caption{Section of a $3$-term grouping sheaf modeling a $3$-deep queue}
\label{fig:grouping_section}
\end{center}
\end{figure}

In addition to retaining a queue of packets, nodes also need to be able to switch between transmit, receive, and idle states.  These states control two behaviors (1) the link activity and (2) when the queues advance.  (For the moment, let us only consider the transmit queue in a node.)  The simplest way to manage several models of behavior is by using a state variable to switch between them.  This is well established in the mathematical systems theory literature \cite{Mosleh_1994,Laskey_1996}, though we merely leverage the ideas as needed here without going into more detail.  

Consider the sheaf model shown in Figure \ref{fig:model_uncert}, which exhibits switching behavior.  This sheaf represents two inputs (on the left) valued in a set $A$ linking to a single $A$-valued output (on the right).  The data over the vertex contains a boolean state value and both inputs.  The left two restriction maps merely extract the two inputs.  The right restriction map is given by the expression $zy+(1-z)x$.  Observe that when $z=0$, this map produces $x$.  On the other hand, if $z=1$ the map produces $y$. 

\begin{figure}
\begin{center}
\includegraphics[width=2in]{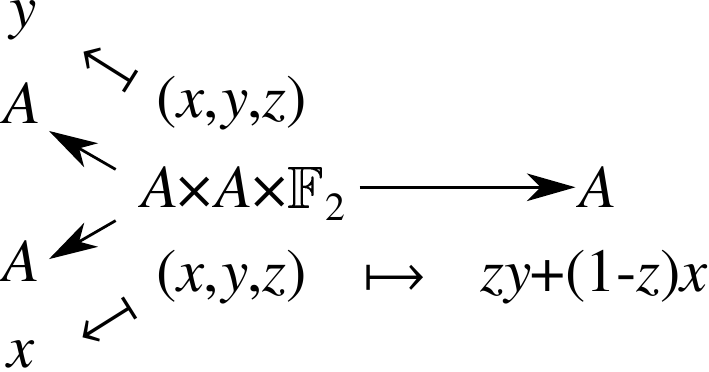}
\caption{A sheaf that models switching an output (right) between two inputs (left)}
\label{fig:model_uncert}
\end{center}
\end{figure}

The example in Figure \ref{fig:model_uncert} is easily generalized to use more elaborate state variables and to exhibit more complex behavior.  Consider a wireless network with two nodes $N=\{0,1\}$ over the course of two consecutive timeslices.  The sheaf we will construct is summarized in Figure \ref{fig:two_nodes_timeseries}, which is written over a time-dependent link complex.  It is helpful to examine the horizontal restriction maps (within a single node) separately from the vertical restrictions (within a single timeslice).   

\begin{figure}
\begin{center}
\includegraphics[width=3.5in]{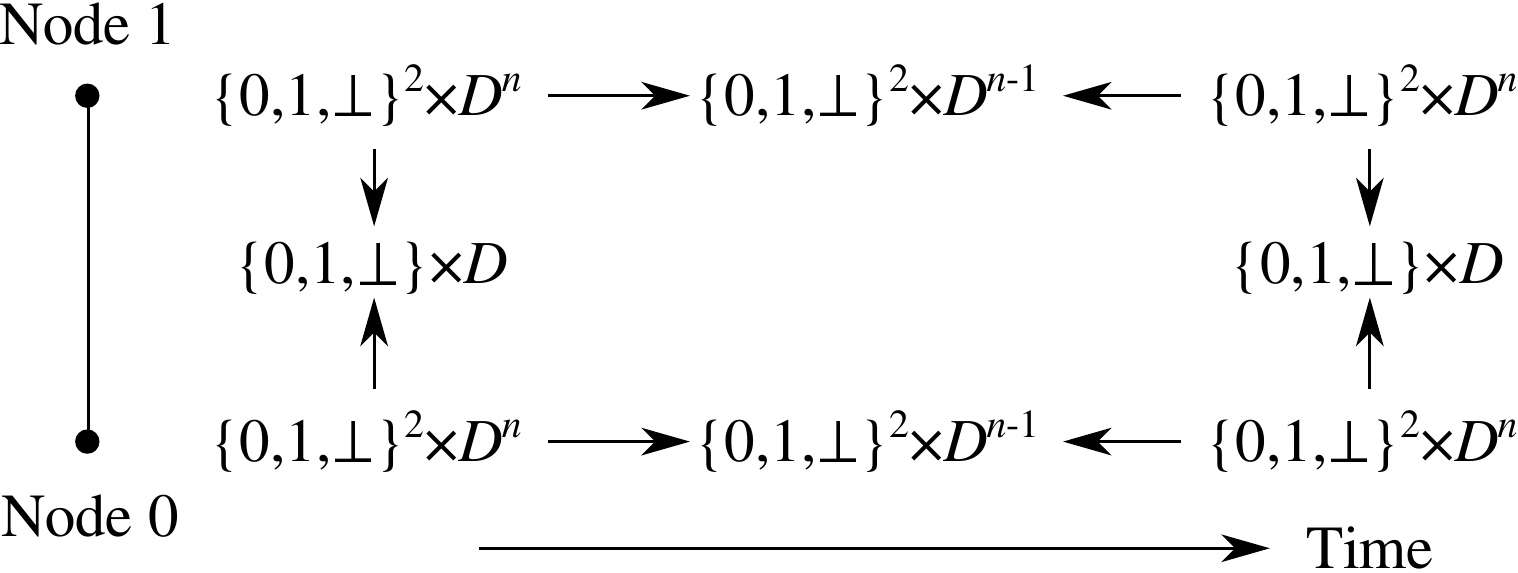}
\caption{A sheaf representing two timesteps of two nodes with a shared channel}
\label{fig:two_nodes_timeseries}
\end{center}
\end{figure}

The diagram shown in Figure \ref{fig:two_nodes_timeseries_time} represents the state of a single node in two timesteps.  The node contains two state variables, a receive buffer, and a transmit queue.  The state variables take values in $N \cup \{\perp\}$ where $N$ is the set of nodes in the network, which represents the identity of the node which is currently transmitting or $\perp$ if the link is idle.  The state variables lie in a $2$-deep queue, so the stalk at a vertex contains both the current and the previous state.

\begin{figure}
\begin{center}
\includegraphics[width=4in]{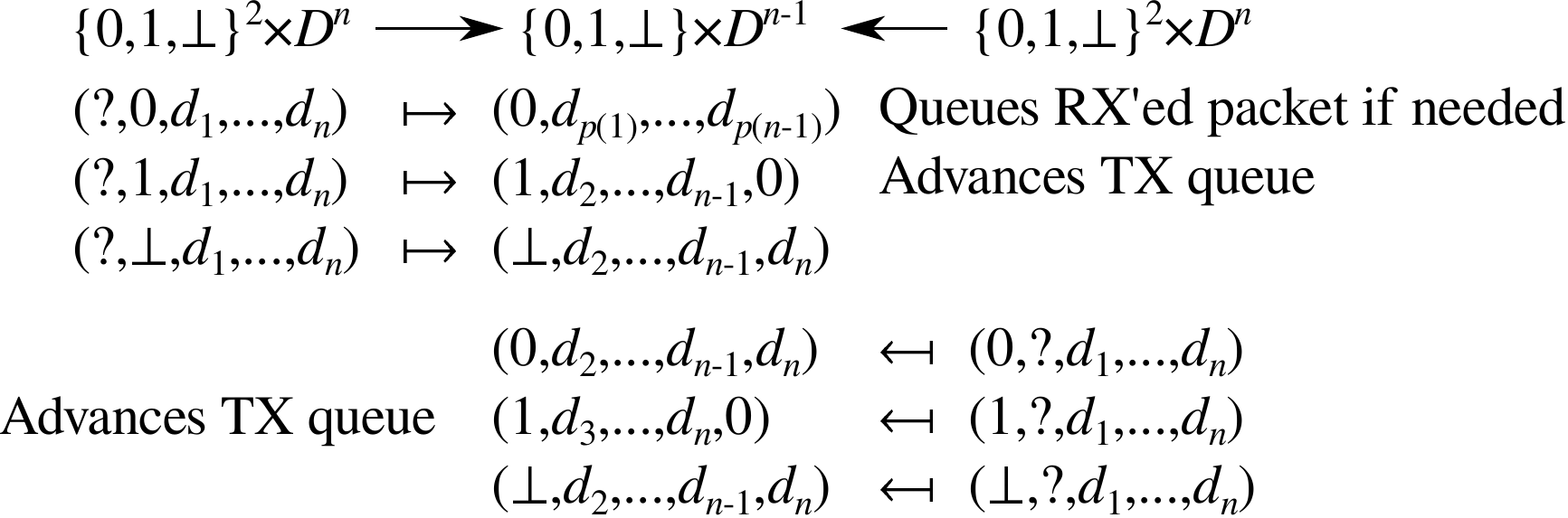}
\caption{The between-timestep restrictions for the sheaf in Figure \ref{fig:two_nodes_timeseries}}
\label{fig:two_nodes_timeseries_time}
\end{center}
\end{figure}

The receive buffer and transmit queue are encoded in the vector $D^n$: the receive buffer is the first element, and the transmit queue is $n-1$ elements long.  Unlike a grouping sheaf, which always advances from one timestep to the next, the transmit queue is only advanced when the node is transmitting.  Thus, the stalk over the edge in Figure \ref{fig:two_nodes_timeseries_time} needs to accomodate a vector that is the same length as the transmit queue to retain the entire queue if the node is not transmitting.  When tranmission occurs, the final element of the vector over the edge is set to zero, effectively yielding a grouping sheaf construction.

When the node is receiving, it must make a determination about whether to forward the received packet or not.  This decision needs to take into account at least the received packet's destination and priority.  
\begin{df}
To represent this decision process abstractly -- without committing to a specific protocol -- this should be modeled as a \emph{receive queue function} $q:D^n \to D^{n-1}$ of the form
\begin{equation*}
q(x_1,x_2,\dotsc,x_n)=\left(x_{p(1,x_1,\dotsc,x_n)},x_{p(2,x_1,\dotsc,x_n)},\dotsc,x_{p(n-1,x_1,\dotsc,x_n)}\right)
\end{equation*}
where for brevity, $p(i,x_1,\dotsc,x_n)$ is written $p(i)$ in Figure \ref{fig:two_nodes_timeseries_time}.  
\end{df}

Notice that a receive queue function $q$ is a permutation of a subset of its inputs that depends on what those inputs are.  Different choices for $q$ determine different protocols for forwarding packets, which usually make decisions based on the contents of $x_1$.  For instance, the following are some possible receive queue functions of a node $n_i \in N$ (this list should not be considered exhaustive, nor are any of these protocols optimal in any useful sense)
\begin{description}
\item[Forward nothing:] 
\begin{equation*}
q(x_1,\dotsc,x_n)=(x_2,x_3,\dotsc,x_n).
\end{equation*}  This protocol doesn't forward any received packets -- it is assumed that the transmit queues are filled internally by the node.
\item[Forward everything:] 
\begin{equation*}
q(x_1,\dotsc,x_n)=(x_1,x_3,\dotsc,x_n).
\end{equation*}
 This protocol places every received packet at the start of the transmit queue, overwriting whatever happened to be there.  Because of this, the protocol may drop packets at the start of the transmit queue.  It is helpful to assume that $x_i=0$ means that nothing is waiting in that slot of the queue.
\item[Forward everything, with queue management:] 
\begin{equation*}
q(x_1,\dotsc,x_n)=(0,\dotsc,0,x_1,\dotsc,x_n),
\end{equation*}
where $x_1$ is placed in the last empty slot in the queue so that it gets transmitted sooner than in the previous protocol.  If there are no nonzero slots in the queue, then the starting queue entry is dropped and replaced with $x_1$.
\item[Forward everything for other recipients:] 
\begin{equation*}
q(x_1,\dotsc,x_n)=\begin{cases}
(0,\dotsc,0,x_1,\dotsc,x_n)&\text{if destination of }x_1 \text{ is not }n_i\\
(x_2,x_3,\dotsc,x_n)&\text{if destination of }x_1 \text{ is }n_i\\
\end{cases}
\end{equation*}
  This protocol forwards every packet that's not destined for the node that received it, by placing it at the start of the queue.  Packets that are destined for the receiving node are not forwarded, so it is effectively a combination of the previous two protocols.  Again, the protocol may drop packets from the start of the transmit queue.
\item[Forward everything for other recipients (two priority levels):]
\begin{equation*}
q(x_1,\dotsc,x_n)=\begin{cases}
(x_3,\dotsc,x_n,x_1)&\text{if destination of }x_1 \text{ is not }n_i,\\
&\text{and priority of }x_1 \text{ is HIGH}\\
(0,\dotsc,0,x_1,\dotsc,x_n)&\text{if destination of }x_1 \text{ is not }n_i,\\
&\text{and priority of }x_1\text{ is LOW}\\
(x_2,x_3,\dotsc,x_n)&\text{if destination of }x_1 \text{ is }n_i\\
\end{cases}
\end{equation*}
 This protocol is the same as the previous one, however packets marked as HIGH priority jump to the end of the transmit queue (next to transmit), dropping any packet at the start of the transmit queue.
\end{description}

\begin{figure}
\begin{center}
\includegraphics[width=3.5in]{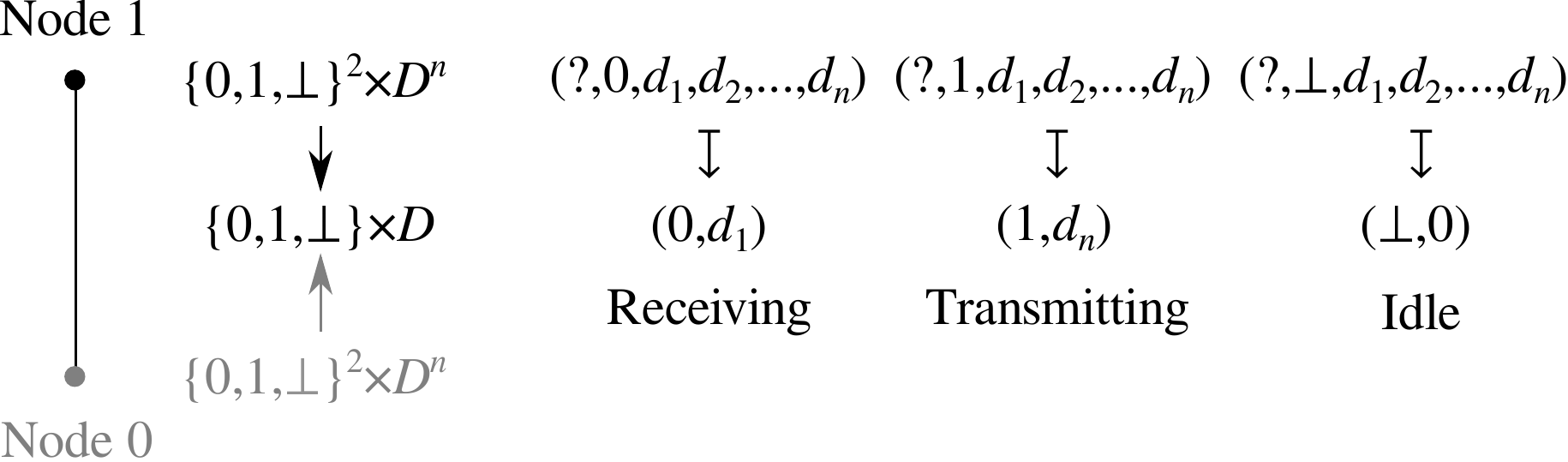}
\caption{The within-timestep restrictions for the sheaf in Figure \ref{fig:two_nodes_timeseries}}
\label{fig:two_nodes_timeseries_link}
\end{center}
\end{figure}

Within a timestep, the restriction maps are based on those of an activation sheaf.  As Figure \ref{fig:two_nodes_timeseries_link} shows, the current state variable determines not only the identity of the node utilizing the link (if any) but also the data present on that link.  Observe that when a node is transmitting, the last element of the transmit buffer is what appears on the link.  When the node is receiving, the link's contents match with the receive buffer. 

Using the above example as a recipe, it is straightforward to define a similar sheaf over a time-dependent link complex.

\begin{df}
A \emph{data payload sheaf} $\shf{D}$ over a time-dependent link complex $L$ with nodes $N$ is parameterized by
\begin{enumerate}
\item A vector space $D$ of possible packets, and
\item A transmit queue length $n-1$.
\end{enumerate}
The stalks of $\shf{D}$ are given by
\begin{description}
\item[For each vertex $c$ of $L$:] $(\{a \in N : \text{there exists a cell } d\in X \text{ with } c \subset d \text{ and } a\in d\}\cup\{\perp\})^2 \times D^n$
\item[For each edge of the form $((c,t),(c,t+1))$:] $(\{a \in N : \text{there exists a cell } d\in X \text{ with } c \subset d \text{ and } a\in d\}\cup\{\perp\}) \times D^{n-1}$
\item[For all other simplices $c$ of $L$:] $(\{a \in N : \text{there exists a cell } d\in X \text{ with } c \subset d \text{ and } a\in d\}\cup\{\perp\}) \times D$
\end{description}
The restrictions of $\shf{D}$ are given by
\begin{enumerate}
\item Between timeslices
\begin{equation*}
\shf{D}\left((a,t) \subset ((a,t),(a,t+1))\right)(n_1,n_2,x_1,\dotsc,x_n)=  
\begin{cases}
(n_2,x_2,\dotsc,x_{n-1},0)&\text{if }n_2=a\text{ and } x_n\not= 0\\ 
(n_2,x_{p(1)},\dotsc,x_{p(n-1)})&\text{if } n_2\not= a\text{ and }n_2\not= \perp\\ 
(\perp,x_2,\dotsc,x_n)&\text{otherwise}\\ 
\end{cases}
\end{equation*}
where $(x_1,\dotsc,x_n) \mapsto (x_{p(1)},\dotsc,x_{p(n-1)})$ is a receive queue function.
\begin{equation*}
\shf{D}\left((a,t+1) \subset ((a,t),(a,t+1))\right)(n_1,n_2,x_1,\dotsc,x_n)= 
\begin{cases}
(n_1,x_3,\dotsc,x_n,0)&\text{if }n_1=a\\ 
(n_1,x_2,\dotsc,x_n)&\text{if }n_1\not= a\text{ and }n_1\not= \perp\\ 
(\perp,x_2,\dotsc,x_n)&\text{otherwise}\\ 
\end{cases}
\end{equation*}

\item Within timeslices, all restrictions between simplices $a \subset b$ of dimension 1 or higher are of the form
\begin{equation*}
\shf{D}\left((a,t) \subset (b,t))\right)(n,x)=
\begin{cases}
(n,x)&\text{if }(n,x) \in \shf{D}((b,t))\\
(\perp,0)&\text{otherwise}\\
\end{cases}
\end{equation*}
while restrictions from a vertex $a$ to an edge $(a,b)$ are given by
\begin{equation*}
\shf{D}\left((a,t) \subset ((a,t),(b,t))\right)(n_1,n_2,x_1,\dotsc,x_n)=
\begin{cases}
(n_2,x_n)&\text{if }n_2= a\text{ and } x_n\not= 0\\ 
(n_2,x_1)&\text{if }n_2\not= a\\ 
(\perp,0)&\text{otherwise}\\ 
\end{cases}
\end{equation*}
\end{enumerate}
\end{df}

\begin{proposition}
Every data payload sheaf contains an activation sheaf as a subsheaf when restricted to any timeslice.
\end{proposition}
\begin{proof}
It is only necessary to match the definition of the stalks and restrictions within a timeslice, and to project onto the second component of each stalk over vertices and the first component over all other simplices.
\end{proof}

This means that the data payload sheaf incorporates the transmission structure described previously for activation sheaves, and more importantly that the within-timeslice restrictions for the data payload sheaf describe the relationship between the data on links and within the nodes.  

For clarity, if $\shf{D}$ is a data payload sheaf on a time-dependent link complex $X$, the activation subsheaf at time $t$ is written $A_t\shf{D}$.  Therefore, there is a collection of surjections on stalks $A_t(a):\shf{D}(a,t) \to A_t\shf{D}(a)$ that project out the appropriate components of the stalks as in the proof above.  These surjections have the property that they are compatible with both sheaves, in that if $a \subset b$
\begin{equation*}
A_t(b) \circ  \shf{D}(a \subset b) = A_t\shf{D}(a \subset b) \circ A_t(a). 
\end{equation*}
This is taken to be the description of a \emph{sheaf morphism} $A_t : \shf{D} \to A_t \shf{D}$.

\begin{proposition}
When restricted to a single node $n$, every data payload sheaf contains a $2$-grouping sheaf taking values in the nodes adjacent to $n$ as a subsheaf.
\end{proposition}
\begin{proof}
By inspection of the restriction maps in the definitions.
\end{proof}

As a result, transmissions between timeslices are decoupled from one another.  It is important to realize that this does \emph{not} mean that the data payloads are decoupled.   Instead, given a sequence of nodes that transmit at each timeslice -- global sections of each activation sheaf in each timeslice -- the data payload sheaf will describe the pathways for threading data through the network as the next proposition states.  

\begin{proposition}
\label{prop:fixed_activation_capacity}
Given a time-dependent link complex $X$ and a data payload sheaf $\shf{D}$ and global sections $\{s_t\}$ for each activation subsheaf $A_t\shf{D}$, then 
\begin{enumerate}
\item the restriction of each stalk $\shf{D}(a,t)$ to the collection of elements whose image through $A_t(a)$ is $s_t(a)$ yields a subsheaf $\shf{P}$ and
\item $\shf{P}$ is a sheaf taking values in the same category as the data payloads $D$, and
\item if $D$ is a vector space, the dimension of the space of global sections of $\shf{P}$ is an upper bound on the network's throughput given the transmission pattern described by $\{s_t\}$.
\end{enumerate}
\end{proposition}
\begin{proof}
\begin{enumerate}
\item The stalk of $\shf{P}$ over a vertex $a$ at time $t$ is given by $\shf{P}(a,t) = D^n$, the stalk over an edge $e$ connecting two timeslices is $\shf{P}(e,t) = D^{n-1}$, and the stalk over any other simplex is merely $D$.  The restrictions then all become projection maps between copies of $D$.   
\item This is evident by construction, since the portion of each stalk of $\shf{D}$ that remains after fixing the $\{s_t\}$ consists of a product of copies of $D$.
\item Suppose it is possible that a packet $x$ can transmitted from node $a$ to node $b$ with the activation pattern given by $\{s_t\}$.  To prove the bound, we construct a global section of $\shf{P}$ that supports this particular data transfer.  This means there is a sequence of nodes $a=n_1, n_2, \dotsc, n_k=b$ through which the packet travels, and that $n_i$ and $n_{i+1}$ are adjacent nodes for each $i$.  This means that those nodes must come active in that sequence, so there is a sequence $t_1 < \dotsc < t_k$ in which $s_{t_i}(n_i)=n_i$ for all $i$ and for which the data over the star of $n_i$ at that time $t_i$ is $x$ for all $i$.  This procedure defines a local section $S$ of $\shf{P}$ on the vertices $(n_i,t_i)$, which can easily be extended to be a section over all timeslices by assigning zero data to all simplices outside the region of influence of the active nodes.  It remains to extend to the edges connecting timeslices, but this is merely a matter of assigning zeros to the yet-unassigned slots in the queues.  
\end{enumerate}
\end{proof}

Beware that the upper bound given in Proposition \ref{prop:fixed_activation_capacity} is not tight -- there may be global sections that describe packets that do not reach their intended destination(s).

\section{Next steps}

The next logical theoretical steps for these analyses are the following:
\begin{enumerate}
\item \emph{Analyzing packet flow through the network.}  We aim for a theorem that shows that packets cannot overlap except if they're \emph{waiting} in a queue at a node.  They cannot overlap on higher dimensional faces within a timeslice.
\item \emph{Protocol level jamming.} Showing how it can be possible to use regions of influence of particular routes to obstruct the network, possibly at particular times.
\item \emph{Developing protocol models through a layered construction.}  The resulting construction will likely leverage a filtration of sheaves or more generally a sequence of sheaf morphisms that describe increasing protocol complexity.  In either of these settings, the abstract tool of spectral sequences may describe constraints imposed on higher layers from the lower ones.
\end{enumerate}

\section*{Acknowledgement}
This research was developed with funding from the Defense Advanced Research Projects Agency (DARPA) via Federal contract HR0011-15-C-0050.  The views, opinions, and/or findings expressed are those of the authors and should not be interpreted as representing the official views or policies of the Department of Defense or the U.S. Government.

%







\bibliographystyle{IEEEtran}
\bibliography{storeforward_bib}

\end{document}